\documentclass{amsart} 

\usepackage{amsmath,amssymb}
\usepackage{amsthm}
\usepackage{amsrefs}
\usepackage{indentfirst}
\usepackage{mathrsfs}
\usepackage{graphicx}
\usepackage{hyperref}
\usepackage{xcolor}
\usepackage{fourier}
\usepackage[margin=1.5in]{geometry}
\usepackage{enumerate}

\theoremstyle{plain}
\newtheorem{theorem}{Theorem}

\newtheorem{prop}[theorem]{Proposition}

\theoremstyle{definition}

\theoremstyle{remark}
\newtheorem*{remark}{Remark}

\DeclareMathOperator{\divop}{div}

\newcommand{\ud}{\,\mathrm{d}}
\newcommand{\rd}{\mathrm{d}}

\newcommand{\RR}{\mathbb{R}}

\newcommand{\PP}{\mathbb{P}}

\newcommand{\dps}{\displaystyle}

\newcommand{\bd}[1]{\boldsymbol{#1}}
\newcommand{\wt}[1]{\widetilde{#1}}
\newcommand{\wh}[1]{\widehat{#1}}

\IfFileExists{mathabx.sty}%
  {\DeclareFontFamily{U}{mathx}{\hyphenchar\font45}%
   \DeclareFontShape{U}{mathx}{m}{n}{<->mathx10}{}%
   \DeclareSymbolFont{mathx}{U}{mathx}{m}{n}%
   \DeclareFontSubstitution{U}{mathx}{m}{n}%
   \DeclareMathAccent{\widebar}{0}{mathx}{"73}%
}{%
  \PackageWarning{mathabx}{%
    Package mathabx not available, therefore\MessageBreak substituting
    widebar with overline\MessageBreak }%
  \newcommand{\widebar}[1]{\overline{#1}}%
}
\newcommand{\wb}[1]{\widebar{#1}}

\newcommand{\mc}[1]{\mathcal{#1}}

\newcommand{\abs}[1]{\lvert#1\rvert}

\newcommand{\average}[1]{\left\langle#1\right\rangle}

\newcommand{\jump}{\text{jump}}

\title{Methodological and computational aspects of parallel tempering methods
  in the infinite swapping limit}

\author{Jianfeng Lu} 
\address{Department of Mathematics, Department of Physics, and
  Department of Chemistry, Duke University, Box 90320, Durham NC
  27708, USA}
\email{jianfeng@math.duke.edu} 

\author{Eric Vanden-Eijnden} 
\address{Courant Institute of
  Mathematical Sciences, New York University, 251 Mercer Street, New
  York, NY 10012, USA}
 \email{eve2@cims.nyu.edu} 

\date{December 19, 2017}

\begin{document}

\begin{abstract}
  A variant of the parallel tempering method is proposed in terms of a
  stochastic switching process for the coupled dynamics of replica
  configuration and temperature permutation.  This formulation is
  shown to facilitate the analysis of the convergence properties of
  parallel tempering by large deviation theory, which indicates that
  the method should be operated in the infinite swapping limit to
  maximize sampling efficiency. The effective equation for the replica
  alone that arises in this infinite swapping limit simply involves
  replacing the original potential by a mixture potential. The
  analysis of the geometric properties of this potential offers a new
  perspective on the issues of how to choose of temperature ladder,
  and why many temperatures should typically be introduced to boost
  the sampling efficiency. It is also shown how to simulate the
  effective equation in this many temperature regime using multiscale
  integrators. Finally, similar ideas are also used to discuss
  extensions of the infinite swapping limits to the technique of
  simulated tempering.
\end{abstract}

\maketitle

\section{Introduction}
\label{sec:intro}

Techniques such as parallel \cite{Sugita1999, Cecchini2004, Bussi2006}
and simulated tempering \cite{Marinari1992} have become standard tools
to accelerate the sampling of systems with complicated energy
landscapes, mostly because of their effectiveness and minimal
requirement of detailed knowledge of the system's specifics. In a
nutshell, the basic idea of these methods is to temporarily raise the
temperature of the system (or that of a copy thereof) to allow it to
explore its energy landscape faster.  When this is done in a specific
way, the properties of the system at the physical temperature can
still be calculated without bias. As a result tempering techniques
provide us with importance sampling schemes to analyze the statistical
mechanics properties of a system at a given temperature even in
situations where direct simulations at this temperature are too slow
to be practical. The aim of this paper is to give some mathematical
analysis of these methods via their formulation as stochastic
switching processes, which will permit us to conveniently explore
their properties in various parameter regimes. In so doing, we will
also discuss implementation issues.

Consider first parallel tempering, in which one evolves multiple
copies (or replicas) of the system at different temperatures which are
occasionally swapped among these replicas (which is why the method is
also known as replica exchange). Typically, these swaps are attempted
at a given frequency $\nu$ with a Metropolis-Hastings
acceptance/rejection criterion to guarantee that the replicas sample a
known equilibrium distribution over which expectations at any of the
temperatures in the stack (including the physical one) can be
calculated. In this context it was found empirically
\cite{Sindhikara2008,Rosta2009,Sindhikara2010} that the larger the
swapping frequency, the higher the sampling efficiency is. This
observation was mathematically justified in the work of Dupuis et
al.~\cite{Dupuis2012} by showing that the large deviation rate
functional for the empirical measure is monotonically increasing with
$\nu$. This observation also motivated the development of infinite
swapping replica exchange dynamics \cite{Dupuis2012,Plattner2011,
  Lu2013,Doll2015}, in which one attempts to simulate directly the
limiting dynamics at $\nu \to \infty$.

How to do so in practice is non-trivial, however. When parallel
tempering is applied to complicated systems, often times a large
number ($N \sim 50-100$) of temperatures and replicas is needed to
make the sampling efficient. In such cases, the simulation of the
infinite swappling limit becomes challenging because the limiting
dynamics involves coefficients given in terms of sums with as many
terms as there are permutations of the replicas, $N!$. An exhaustive
evaluation of such sums quickly become undoable, even numerically, and
alternative strategies must be used. One possibility, proposed
in~\cite{TangLuAbramsVE}, is to resort to multiscale integrators like
those in the heterogeneous multiscale methods (HMM) to effectively
simulate the limiting dynamics.

One of the purposes of this paper is to present the mathematical
foundation behind the algorithm proposed
in~\cite{TangLuAbramsVE}. Specifically, we show that the parallel
tempering can be formulated as a diffusion of the replica coupled with
a Markov jumping process in the space of permutations. Together this
gives a stochastic switching process.  Using the approach introduced
in \cite{Dupuis2012}, the infinite swapping limit can be justified by
the monotonicity of the large deviation rate functional as a function
of the swapping frequency. The large swapping frequency introduces a
scale separation between the dynamics of the replica and that of the
permutation variable, and thus naturally calls for multiscale
integrators, such as those based on HMM, which justifies the
integrator introduced in \cite{TangLuAbramsVE}.

Similar ideas can also be used in the context of simulated tempering,
in which a single replica is used but the temperature itself is made a
dynamical variable of the system. In this setup too, it is useful to
consider the coupled dynamics of the system and the temperature as a
stochastic switching process, and this allows one to again justify via
large deviation theory that making the temperature evolve infinitely fast
(the `infinite switch limit') is optimal. As we will show below, the
limiting equation that emerges in this infinite-switch limit is easier
to simulate in principle, as it only involves a sum over the number of
possible temperature states rather than its factorial. In practice,
however, the question becomes how to appropriately weight the
different temperatures states, since this is an input in the method
which is unknown \textit{a~priori} but can greatly influence its
efficiency.

The organization of the remainder of this paper is as follows. In
Section~\ref{sec:remd} we present a formulation of parallel tempering
as a stochastic switching process, which facilitates the discussion of
infinite swapping limit and the development of integrators for the
equations of motion that emerge in this limit.  The infinite swapping
limit is justified through large deviation theory in
Section~\ref{sec:LDP} and the effective equations that arise in this
limit are derived in Section~\ref{sec:infinity}. In
Section~\ref{sec:harmo} we discuss the choice of the temperature
ladder in parallel tempering and, in particular, explain why many
temperatures are typically necessary to boost the efficiency of the
method.  In the infinite swapping limit, where the dynamics of the
replica effectively reduce to a diffusion over a mixture potential,
this question reduces to the analysis of the geometrical properties of
this potential. An HMM multiscale integrator that operates in the many
temperatures regime is then discussed in Section~\ref{sec:hmm} and
shown to perform better than the standard stochastic simulation
algorithm in the regime of frequent swapping.  In
Section~\ref{sec:simulated} we show that similar ideas can be extended
to other tempering techniques such as simulate tempering. Finally,
some concluding remarks are given in Section~\ref{sec:conclu}.

\section{Parallel Tempering as a Stochastic Switching Process}
\label{sec:remd}

Consider a  system whose evolution
is governed by the overdamped Langevin equation:
\begin{equation}
  \label{eq:overdamp}
   \ud \bd{x}(t) = - \nabla V(\bd{x}(t)) \ud t + \sqrt{2  \beta^{-1}} \ud \bd{W}(t),
\end{equation}
where $\bd{x}(t) \in \Omega \subset \RR^{3n}$ denotes the
instantaneous position (at time $t$) of the system with $n$ particles,
$-\nabla V$ is the force associated with the potential $V$,
$\beta= 1/k_B T$ is the inverse temperature, each component of
$\bd{W}(t) \in \RR^{3n}$ is an independent Wiener process, and we set the
friction coefficient to one for simplicity. Most of the results we
present below should be straightforwardly generalizable to other
dynamics, such as the inertial Langevin equation, but we will focus
on~\eqref{eq:overdamp} for the sake of simplicity.

Assuming that $V$ is smooth and $\Omega$ is compact, it is easy to
show that the overdamped dynamics \eqref{eq:overdamp} is ergodic with
respect to the Boltzmann equilibrium distribution with density
\begin{equation}
 \label{eq:eqpdf}
 \rho_\beta(\bd{x}) = Z_\beta^{-1} e^{-\beta V(\bd{x})},
\end{equation}
where $Z_\beta = \int_{\Omega}e^{-\beta V(\bd{x})} \ud \bd{x}$.
However, when the dimensionality of the system is large, $n\gg1$, and
the potential $V$ is complicated with critical points,
etc. \eqref{eq:overdamp} typically exhibits metastability, meaning
that convergence to the equilibrium measure is very slow. In this
context, the basic idea of parallel tempering is to use dynamics at
higher temperature to help the sampling.

Specifically, we introduce multiple replicas of the system and make
each replica evolve under a different temperature that alternatively
swaps between several of them (many artificial temperatures, plus the
physical one).  Assuming that we use $N$ temperatures, so that the
parallel tempering correspondingly involves $N$ replicas, this is done
by constructing a stochastic process on the extended phase space
$\Omega^N \times P_N$, where $\Omega^N$ is the configurational space
for the $N$ replica and $P_N$ is the permutation group of the $N$
temperature, such that its joint equilibrium measure is given by
\begin{equation}
  \label{eq:14}
  \pi(d\bd{X},\sigma) = \varrho(\bd{X},\sigma) d\bd{X}, \qquad 
  \varrho(\bd{X},\sigma) = \frac1{N!} 
  \prod_{i=1}^N 
  \rho_{\beta_{\sigma_i} }(\bd{x}_i),
\end{equation}
where $\bd{X}=(\bd{x_1},\bd{x_2},\ldots, \bd{x_N})$ and $\sigma_i$
denotes the index associated with $i$ by the permutation $\sigma$.  To
this end, we assume that given any $\sigma \in P_N$ we have a
diffusion with generator $\mc{L}_{\sigma}$ whose equilibrium
density is the product
\begin{equation}
  \label{eq:13}
  \varrho(\bd{X}\mid\sigma) = \frac{\varrho(\bd{X},\sigma)}{\int
    _{\Omega^N} \varrho(\bd{X}',\sigma)d\bd{X}'}
  = \prod_{i=1}^N \rho_{\beta_{\sigma_i} }(\bd{x}_i).
\end{equation}
We also assume that given any $\bd{X}\in\Omega^N$, we have a Markov
jump process with infinitesimal jump intensity
$\nu h_{\sigma\sigma'}(\bd{X})$, where $\nu > 0$ is an overall
frequency parameter and $h_{\sigma\sigma'}(\bd{X})\ge 0 $ satisfies the
detailed balance condition
\begin{equation}\label{eq:detailbalance}
  \varrho(\sigma \mid \bd{X}) h_{\sigma\sigma'}(\bd{X}) 
  = \varrho(\sigma' \mid \bd{X}) h_{\sigma'\sigma}(\bd{X}).
\end{equation}
where 
\begin{equation}
  \label{eq:15}
  \varrho(\sigma \mid \bd{X}) =
  \frac{\varrho(\bd{X},\sigma)}{\displaystyle\sum_{\sigma'\in P_N} \varrho(\bd{X},\sigma')}
  = \frac{\displaystyle\prod_{i=1}^N e^{-\beta_{\sigma_i} V(\bd{x}_i) }}
  {\displaystyle\sum_{\sigma'\in P_N} \prod_{i=1}^N e^{-\beta_{\sigma'_i} V(\bd{x}_i)}}.
\end{equation}
The generator for the stochastic switching process
$(\bd{X}(t), \sigma(t))\in \Omega^N \times P_N $ is then given by
\begin{equation}
  \label{eq:generator}
  (\mc{L}_{\nu} u)(\bd{X}, \sigma) = \mc{L}_{\sigma} u(\bd{X}, \sigma)
  -  \nu \sum_{\sigma' \neq \sigma} h_{\sigma\sigma'}(\bd{X})  \bigl( u(\bd{X}, \sigma) 
  - u(\bd{X}, \sigma') \bigr), 
\end{equation}
where the subscript emphasizes the dependence of $\mc{L}_{\nu}$ on the
overall swapping frequency $\nu$.
By construction, any stochastic switch processes constructed this way
have $\varrho(\bd{X},\sigma) $ as the invariant measure.
\begin{prop}\label{prop:1}
  For any attempt switching frequency $\nu > 0$, $\pi(d\bd{X},\sigma)$ is the
  invariant distribution of the process associated with $\mc{L}_{\nu}$.
\end{prop}

\begin{proof}
  It suffices to show that 
  \begin{equation}
    \mc{L}_{\nu}^{\ast} \varrho = 0.
  \end{equation}
  This follows from a direct calculation:
  \begin{equation}
    \begin{aligned}
      \Bigl(\mc{L}_{\nu}^{\ast} \varrho \Bigr)(\bd{X}, \sigma) 
      & =  \mc{L}_\sigma^{\ast} \varrho(\bd{X}, \sigma) - \nu
      \sum_{\sigma'} h_{\sigma\sigma'}(\bd{X}) 
      \varrho(\bd{X}, \sigma) + \nu \sum_{\sigma'} 
      h_{\sigma'\sigma}(\bd{X}) \varrho(\bd{X}, \sigma') \\
      & = \Bigl( \mc{L}_{\sigma}^{\ast} \varrho(\bd{X} \mid \sigma)
      \Bigr) \varrho(\sigma) + \Bigl( - \nu \sum_{\sigma'}
      h_{\sigma\sigma'}(\bd{X}) \varrho( \sigma\mid \bd{X}) + \nu
      \sum_{\sigma'} h_{\sigma'\sigma}(\bd{X}) \varrho( \sigma'\mid
      \bd{X})
      \Bigr) \varrho(\bd{X}) \\
      & = 0,
    \end{aligned}
  \end{equation}
  where the last equality uses the assumption that
  $\varrho(\bd{X} \mid \sigma)$ is stationary with respect to the
  dynamics corresponds to $\mc{L}_{\sigma}$ and
  $h_{\sigma\sigma'}(\bd{X})$ satisfies the detailed balance condition
  \eqref{eq:detailbalance}.
\end{proof}

It is easy to see that we can write
$\varrho(\bd{X}, \sigma) = \bigl(N!\prod_{i=1}^N Z_{\beta_i}\bigr)^{-1}
\exp\bigl(-\beta \mc{V}(\bd{X}, \sigma)\bigr)$
by defining the potential
\begin{equation}
  \mc{V}(\bd{X}, \sigma) = \beta^{-1} \sum_{i=1}^N \beta_{\sigma_i} V(\bd{x}_i), 
\end{equation}
where $\beta $ is a reference inverse temperature introduced for
dimensional purpose: for example, we can take $\beta = \beta_{1}$.
Similarly, the conditional density~\eqref{eq:13} can be written in
terms of $\mc{V}(\bd{X}, \sigma)$ as
\begin{equation}
  \varrho(\bd{X} \mid \sigma) = \left(\prod_{i=1}^N
    Z_{\beta_i}^{-1}\right) 
  \exp\bigl(-\beta \mc{V}(\bd{X}, \sigma) \bigr), 
\end{equation}
which can be sampled by, for example, by the diffusion
\begin{equation}
  \begin{aligned}
  \ud \bd{x}_i & = - \nabla_{\bd{x}_i} \mc{V}(\bd{X}, \sigma) + 
  \sqrt{2 \beta^{-1}} \ud \bd{W}^{(i)}_t \\
  & = - \beta^{-1} \beta_{\sigma_i} \nabla V(\bd{x}_i) 
  + \sqrt{2 \beta^{-1}} \ud \bd{W}^{(i)}_t, \quad i = 1, \ldots, N.
  \end{aligned}
\end{equation}
We can also re-express~\eqref{eq:15} as
\begin{equation}
  \varrho(\sigma \mid \bd{X}) = 
  \dfrac{e^{-\beta \mc{V}(\bd{X}, \sigma)}}{\sum_{\sigma'\in P_N} e^{-\beta \mc{V}(\bd{X}, \sigma')}}
\end{equation}
Thus, to ensure the detailed balance, we may choose $h$ as 
\begin{equation}\label{eq:jumpProb}
  h_{\sigma, \sigma'}(\bd{X}) = a_{\sigma, \sigma'} e^{-\frac{1}{2} \beta (\mc{V}(\bd{X}, \sigma') - \mc{V}(\bd{X}, \sigma))},
\end{equation}
where the symmetric adjacency matrix
$a_{\sigma, \sigma'} = a_{\sigma', \sigma} \in\{0, 1\} $ indicates
whether the permutation $\sigma'$ is allowed to be accessed from
$\sigma$ by a single jump. For example, we may restrict ourselves to
transposition moves in which two random indices are swapped from the
permutation $\sigma$, so that $N \choose 2$ permutations $\sigma'$ are
accessible from $\sigma$. Note that the choice \eqref{eq:jumpProb} is
not unique, and $h_{\sigma,\sigma'}$ can also take other forms that
satisfy the detailed balance condition with respect to
$\varrho(\sigma \mid \bd{X})$.

\subsection*{Calculation of Expectations}
To estimate the expectation of the physical observable $A(x)$ at any
temperature $\beta_k$, $k=1,\ldots, N$, we can use:
\begin{equation}
  \label{eq:expectation}
  \begin{aligned}
  \average{A}_{\beta_k} & = \int_{\Omega}
  A(\bd{x})\rho_{\beta_k}(\bd{x}) \ud\bd{x}\\
  & =  \sum_{j=1}^N \sum_{\sigma \in P_N}\int_{\Omega^N}
  A(\bd{x}_j) 
  \mathbb{1}_{\sigma_j= k} \, \varrho(\sigma  \mid
  \bd{X}) 
  \,  \varrho(\bd{X}) \ud\bd{X}  \\
  & = \sum_{j=1}^N  \int_{\Omega^N}  A(\bd{x}_j) 
  \eta_{j, k}(\bd{X}) \varrho(\bd{X}) \ud \bd{X}.
  \end{aligned}
\end{equation}
Here $\varrho(\bd{X})  $ is the marginal of $\varrho(\bd{X},\sigma)$  on
$\bd{X}$,
\begin{equation}
  \label{eq:16}
  \varrho(\bd{X})   = \sum_{\sigma\in
      P_N} \varrho(\bd{X},\sigma) = \frac1{N!}\sum_{\sigma \in P_N}\prod_{i=1}^N 
  \rho_{\beta_{\sigma_i} }(\bd{x}_i),
\end{equation}
and we defined
\begin{equation}
  \label{eq:etadef}
  \eta_{j, k}(\bd{X}) = \sum_{\sigma \in P_N} \mathbb{1}_{\sigma_j = k}\, 
  \varrho(\sigma  \mid \bd{X}) 
  = \dfrac{\sum_{\sigma \in P_N}
    \mathbb{1}_{\sigma_j = k}  e^{-\beta \mc{V}(\bd{X}, \sigma)}}{\sum_{\sigma\in P_N} e^{-\beta \mc{V}(\bd{X}, \sigma)}},
\end{equation}
which is the probability that the $j$-th replica is at the $k$-th
temperature conditional on the configuration being fixed at
$\bd{X}$. \eqref{eq:expectation} means that we can estimate
$\average{A}_{\beta_k} $ from a sample path of the stochastic
switching process using ergodicity as
\begin{equation}
  \label{eq:17}
  \average{A}_{\beta_k} = \lim_{T\to\infty} \frac1T \int_0^T 
  \sum_{j=1}^N  A(\bd{x}_j(t)) 
  \eta_{j, k}(\bd{X}(t)) \ud t.
\end{equation}

The following proposition shows that the variance of the parallel
tempering estimator for the expectation of an observable at a given
temperature is bounded by the variance of the observable at that
temperature. 
  \begin{prop}\label{prop:variance}
    For each $k = 1, \ldots, N$ we have
    \begin{equation}\label{eq:variance}
      \int_{\Omega^N} \Biggl( \sum_{j=1}^N A(\bd{x}_j) 
      \eta_{j, k}(\bd{X}) \Biggr)^2 
      \varrho(\bd{X}) \ud \bd{X} \leq \int_{\Omega} 
      |A(\bd{x})|^2 \rho_{\beta_k}(\bd{x}) \ud \bd{x}.
    \end{equation}
  \end{prop}
  \begin{proof}
    By definition 
    \begin{equation}
      \sum_{j=1}^N \eta_{j, k}(\bd{X}) = \sum_{j=1}^N \sum_{\sigma \in
        P_N} \mathbb{1}_{\sigma_j = k}\, \varrho(\sigma \mid \bd{X}) 
      = \sum_{\sigma \in P_N} \varrho (\sigma \mid X)  = 1.
    \end{equation}
    Therefore, by Jensen's inequality, we have 
    \begin{equation}
      \Biggl(\sum_{j=1}^N A(\bd{x}_j) \eta_{j, k}(\bd{X})  \Biggr)^2
    \leq \sum_{j=1}^N |A(\bd{x}_j)|^2 \eta_{j, k}(\bd{X}),
  \end{equation}
  and hence
  \begin{equation}
    \begin{aligned}
      \int_{\Omega^N} \Biggl( \sum_{j=1}^N A(\bd{x}_j) \eta_{j,
        k}(\bd{X}) \Biggr)^2 \varrho(\bd{X}) \ud \bd{X} & \leq \int
      \sum_{j=1}^N |A(\bd{x}_j)|^2 \eta_{j, k}(\bd{X}) \varrho(\bd{X})
      \ud \bd{X} = \int_{\Omega} |A(\bd{x})|^2 \rho_{\beta_k}(\bd{x}) \ud \bd{x},
    \end{aligned}
  \end{equation}
  where the last equality follows from \eqref{eq:expectation}. This
  establishes~\eqref{eq:variance}.
\end{proof}

Note that since the bound~\eqref{eq:variance} follows from Jensen's
equality, we expect it to be sharp for generic observables. Therefore,
the efficiency of the sampling scheme is determined by the convergence
to equilibrium of the parallel tempering scheme, as we will discuss in
the following sections.

\section{Large Deviation Principle for the Empirical Measure of the
  Stochastic Switching Process} 
\label{sec:LDP}

In this section we derive a large deviation principle for the
empirical measure of the stochastic switching process $(\bd{X}(t),
\sigma(t))$ marginalized on $\bd{X}(t)$, that is:
\begin{equation}
  \label{eq:18}
  \pi_T(\rd\bd{X}) = \frac1T \int_0^T \delta _{\bd{X}(t)} (\rd\bd{X}) \ud t.
\end{equation}
By construction of the stochastic switching process its infinitesimal
generator splits into the two parts (see~\eqref{eq:generator}):
\begin{equation}
  \mc{L}_{\nu} = \mc{L}_\sigma + \nu \mc{L}_{\jump}, 
\end{equation}
As a consequence, the large deviation rate functional for the
empirical measure also has an additive structure:

\begin{prop} The empirical measure~\eqref{eq:18} satisfies a large
  deviation principle with rate functional given by
\begin{equation}
  I^{\nu}(\mu) = J_\sigma(\mu) + \nu J_{\jump}(\mu), 
\end{equation}
where $\mu$ is a probability measure on $\Omega^N \times P_N$. In
particular, both $J_{\sigma}$ and $J_{\jump}$ are independent of
$\nu$. While the specific form of $J_\sigma$ depends on the choice of
$\mc{L}_{\sigma}$, $J_{\jump}$ is fully determined by the jump
intensity $h_{\sigma\sigma'}(\bd{X})$ as:
\begin{equation}\label{eq:Jjump}
  J_{\jump}(\mu) = \frac{1}{2} \sum_{\sigma, \sigma' \in P_N}
  \int_{\Omega^N} h_{\sigma\sigma'}(\bd{X}) \left[1 - 
    \sqrt{\frac{\theta(\bd{X}, \sigma')}{\theta(\bd{X}, \sigma)}} \right]^2 \mu(\rd\bd{X}, \sigma), 
\end{equation}
where we have assumed that $\mu$ is absolutely continuous with respect
to $\pi$ and denote 
$\theta(\bd{X}, \sigma) := [\rd \mu / \rd \pi](\bd{X}, \sigma)$. 
\end{prop}

The additive structure was first observed in the work of Dupuis et
al. \cite{Dupuis2012, Plattner2011} in the context of parallel
tempering. 
\begin{proof}
  The additive structure comes from the following representation of
  the rate functional \cite{DonskerVaradhan:75, DupuisLiu:2015}:
\begin{equation}
 \begin{aligned}
  I^{\nu}(\mu) & = - \sum_{\sigma \in P_N} \int_{\Omega^N}
  \theta(\bd{X}, \sigma)^{1/2} 
  \mc{L}_{\nu} \bigl(\theta(\bd{X}, \sigma)^{1/2}\bigr) \pi(\rd\bd{X}, \sigma)  \\
  & = - \sum_{\sigma \in P_N} \int_{\Omega^N}  \theta(\bd{X}, \sigma)^{1/2}
  \mc{L}_\sigma \bigl(\theta(\bd{X}, \sigma)^{1/2}\bigr) \pi(\rd\bd{X},
  \sigma)\\
  & \quad - \nu \sum_{\sigma \in P_N} \int_{\Omega^N}  \theta(\bd{X},
  \sigma)^{1/2} 
  \mc{L}_{\jump} \bigl(\theta(\bd{X}, \sigma)^{1/2}\bigr) \pi(\rd\bd{X}, \sigma). 
  \end{aligned}
\end{equation}
Using the definition of $\mc{L}_{\jump}$, we calculate 
\begin{multline}
  - \sum_{\sigma \in P_N} \int_{\Omega^N} \theta(\bd{X}, \sigma)^{1/2}
  \mc{L}_{\jump} \bigl(\theta(\bd{X}, \sigma)^{1/2}\bigr)
  \pi(\rd\bd{X}, \sigma) \\
  = \sum_{\sigma, \sigma'\in P_N} \int_{\Omega^N} \theta(\bd{X},
  \sigma)^{1/2} h_{\sigma\sigma'}(\bd{X}) \bigl(\theta(\bd{X},
  \sigma)^{1/2} - \theta(\bd{X}, \sigma')^{1/2} \bigr) \varrho(\sigma
  \mid\bd{X}) \varrho(\bd{X}) \rd\bd{X}.
  \end{multline}
Since $h_{\sigma\sigma'}(\bd{X})$ satisfies the detailed balance condition, we have 
\begin{equation}
  \begin{aligned}
    &\sum_{\sigma, \sigma'\in P_N} \int_{\Omega^N} \theta(\bd{X}, \sigma)^{1/2}
    h_{\sigma\sigma'}(\bd{X}) \bigl(\theta(\bd{X}, \sigma)^{1/2} -
    \theta(\bd{X}, \sigma')^{1/2} \bigr) \varrho(\sigma
    \mid \bd{X}) \varrho(\bd{X}) \rd\bd{X} = \\
   & = \sum_{\sigma, \sigma'\in P_N} \int_{\Omega^N} \theta(\bd{X}, \sigma)^{1/2}
    h_{\sigma'\sigma}(\bd{X}) \bigl(\theta(\bd{X}, \sigma)^{1/2} -
    \theta(\bd{X}, \sigma')^{1/2} \bigr) \varrho(\sigma'
    \mid \bd{X}) \varrho(\bd{X}) \rd\bd{X}\\
    & = \sum_{\sigma, \sigma'\in P_N} \int_{\Omega^N} \theta(\bd{X}, \sigma')^{1/2}
    h_{\sigma\sigma'}(\bd{X}) \bigl(\theta(\bd{X}, \sigma')^{1/2} -
    \theta(\bd{X}, \sigma)^{1/2} \bigr) \varrho(\sigma \mid \bd{X})
    \varrho(\bd{X}) \rd\bd{X},
  \end{aligned}
\end{equation}
where we swapped the indices $\sigma$ and $\sigma'$ to get the second
equality. Combined with the previous identity, we get
\begin{equation}
  \begin{aligned}
    &- \sum_{\sigma \in P_N} \int_{\Omega^N} \theta(\bd{X}, \sigma)^{1/2}
    \mc{L}_{\jump}
    \bigl(\theta(\bd{X}, \sigma)^{1/2}\bigr) \pi(\rd\bd{X}, \sigma) \\
    & = \frac{1}{2} \sum_{\sigma, \sigma'\in P_N} \int_{\Omega^N}
    h_{\sigma\sigma'}(\bd{X})
    \Bigl[\theta(\bd{X}, \sigma)^{1/2} - \theta(\bd{X}, \sigma')^{1/2}
    \Bigr]^2 
    \pi(\rd\bd{X}, \sigma) \\
    & = \frac{1}{2} \sum_{\sigma, \sigma'\in P_N} \int_{\Omega^N}
    h_{\sigma\sigma'}(\bd{X}) \Bigl[1 - \theta(\bd{X}, \sigma')^{1/2}
    \theta(\bd{X}, \sigma)^{-1/2} \Bigr]^2 \mu(\rd\bd{X}, \sigma).
  \end{aligned}
\end{equation}
This is exactly $J_{\jump}$ defined in \eqref{eq:Jjump}. 
\end{proof}

\section{Infinite Swapping Limit}
\label{sec:infinity}

As $J_{\jump}$ is non-negative, it is natural to consider taking the
limit $\nu \to \infty$ to maximize the rate functional, which we refer
to as the infinite swapping limit.  To derive the effective
dynamics that emerges in this limit, consider the forward Kolmogorov
equation for the stochastic switching process $(\bd{X}(t),\sigma(t))$:
\begin{equation}
  \begin{aligned}
    \partial_t \varrho(t,\bd{X}, \sigma) & = \mc{L}_{\nu}^{\ast} \varrho (t,\bd{X}, \sigma) \\
    & = \mc{L}_{\sigma}^{\ast} \rho(t, \bd{X}, \sigma) - \nu \sum_{\sigma'}
    h_{\sigma\sigma'}(\bd{X}) 
    \varrho(t,\bd{X}, \sigma) +
    \nu \sum_{\sigma'} h_{\sigma'\sigma}(\bd{X}) \rho(t, \bd{X}, \sigma'). 
  \end{aligned}
\end{equation}
To take the limit $\nu \to \infty$, let us introduce the ansatz
\begin{equation}
  \varrho(t, \bd{X}, \sigma) = \varrho_0(t, \bd{X}, \sigma) + \nu^{-1} 
  \varrho_1(t, \bd{X}, \sigma) + \nu^{-2} \varrho_2(t, \bd{X}, \sigma) + \cdots. 
\end{equation}
Matching orders of $\nu$, we get 
\begin{align}
  & \mc{L}_{\jump}^{\ast} \varrho_0 = 0; \label{eq:rho0} \\
  & \mc{L}_{\jump}^{\ast} \varrho_1 = \partial_t \varrho_0 - \mc{L}_\sigma^{\ast}
  \varrho_0, \label{eq:rho1}
\end{align}
and similarly for the higher order expansions. The leading order \eqref{eq:rho0} reads 
\begin{equation}
  \bigl(\mc{L}_{\jump}^{\ast} \varrho_0\bigr)(t, \bd{X}, \sigma) 
  = - \sum_{\sigma'} h_{\sigma\sigma'}(\bd{X}) \varrho_0(t, \bd{X},  \sigma) 
  + \sum_{\sigma'} h_{\sigma'\sigma}(\bd{X}) \varrho_0(t, \bd{X}, \sigma') = 0, 
\end{equation}
which implies that, since by construction $\varrho(\sigma \mid \bd{X})$
is the invariant measure of the jumping process associated with
$\bigl\{h_{\sigma\sigma'}(\bd{X})\bigr\}$ 
\begin{equation}
  \varrho_0(t, \bd{X}, \sigma) = f_0(t, \bd{X}) \varrho(\sigma \mid \bd{X}),
\end{equation}
for some $f_0$ to be determined. Substitution this into the next order
gives
\begin{equation}
  \varrho(\sigma \mid  \bd{X}) \partial_t f_0(t, \bd{X}) =
  \mc{L}_\sigma^{\ast} \varrho_0 
  - \sum_{\sigma'} h_{\sigma\sigma'}(\bd{X}) \varrho_1(t, \bd{X},
  \sigma) 
  + \sum_{\sigma'} h_{\sigma'\sigma}(\bd{X}) \varrho_1(t, \bd{X}, \sigma'). 
 \end{equation}
Summing over $\sigma$, we obtain the following solvability condition
for this equation (recall that $\sum_{\sigma} \varrho(\sigma \mid \bd{X}) = 1$)
\begin{equation}
  \partial_t f_0(t, \bd{X}) = \sum_{\sigma} \mc{L}_{\sigma}^{\ast} 
  \bigl( \varrho(\sigma \mid \bd{X}) f_0(t, \bd{X}) \bigr) =: \wb{\mc{L}}^{\ast} f_0(t, \bd{X}),
\end{equation}
where the last identity defines $\wb{\mc{L}}$, the infinitesimal
process of the ``averaged process'' in the infinite swapping limit.

The above asymptotic derivation can be made rigorous as in usual
averaging, and we conclude:
\begin{prop}
  As $\nu \to \infty$, $\bd{X}(t)$ converges to the averaged process
  defined by the infinitesimal generator $\wb{\mc{L}}$.
\end{prop}

\smallskip

Let us now make these formula explicit.  Since
\begin{equation}
  \mc{L}_{\sigma} = \sum_{k=1}^N \Bigl( - \beta^{-1} \beta_{\sigma_k} \nabla V(\bd{x}_k) \cdot \nabla_{\bd{x}_k} + \beta^{-1} \Delta_{\bd{x}_k} \Bigr),
\end{equation}
we have
\begin{equation}
  \begin{aligned}
    \wb{\mc{L}}^{\ast} f(\bd{X}) & =  \sum_{\sigma\in P_N}
    \mc{L}_{\sigma}^{\ast} 
    \bigl(\varrho(\sigma \mid \bd{X}) f(\bd{X}) \bigr) \\
    & = \sum_{\sigma\in P_N} \sum_{k=1}^N \biggl[\nabla_{\bd{x}_k} \cdot \biggl(
    \bigl(\nabla_{\bd{x}_k} \mc{V}(\bd{X}, \sigma)\bigr)
    \varrho(\sigma \mid \bd{X}) 
    f(\bd{X}) \biggr)  + \beta^{-1}
    \Delta_{\bd{x}_k}\bigl( \varrho(\sigma \mid \bd{X}) 
    f(\bd{X})\bigr) \biggr]
     \\
    & = \sum_{\sigma\in P_N} \sum_{k=1}^N \nabla_{\bd{x}_k} \cdot \biggl(
    \bigl(\nabla_{\bd{x}_k} \mc{V}(\bd{X}, \sigma)\bigr) 
    \varrho(\sigma \mid \bd{X}) f(\bd{X}) \biggr)  +  \beta^{-1} \sum_{k=1}^N
    \Delta_{\bd{x}_k}
    f(\bd{X}),
    \end{aligned}
\end{equation}
where the last equality uses
$\sum_{\sigma} \varrho(\sigma \mid \bd{X}) = 1$ for any $\bd{X}$.
Therefore, the corresponding stochastic differential equations are
given by
\begin{equation}\label{eq:infswap}
  \begin{aligned}
  \ud \bd{x}_j & = - \sum_{\sigma\in P_N} \varrho(\sigma \mid \bd{X})
  \nabla_{\bd{x}_j} 
  \mc{V}(\bd{X}, \sigma) \ud t + \sqrt{2 \beta^{-1}} \ud \bd{W}_t^{(j)}\\
  & = - R_j(\bd{X}) \nabla V(\bd{x}_j) \ud t + \sqrt{2 \beta^{-1}} \ud \bd{W}_t^{(j)}
  \end{aligned}
\end{equation}
with the scaling factor $R_j(\bd{X})$ given by 
\begin{equation}\label{eq:defRj}
  R_j(\bd{X}) = \beta^{-1} \sum_{\sigma\in P_N} \beta_{\sigma_j}
  \varrho(\sigma \mid \bd{X})
  =  \dfrac{\sum_{\sigma \in P_N}
    \beta_{\sigma_j} e^{-\beta \mc{V}(\bd{X}, \sigma)}}{\sum_{\sigma\in P_N} e^{-\beta \mc{V}(\bd{X}, \sigma)}}.
\end{equation}
The dynamics we get after averaging is similar to the original
overdamped dynamics under inverse temperature $\beta$, except that the
drift is scaled by the factor $R_j(\bd{X})$. Notice that in terms of
of the factor $\eta_{j,k}(\bd{X})$ defined in~\eqref{eq:etadef}, we
can write \eqref{eq:defRj} as
s\begin{equation}
  R_j(\bd{X}) = \beta^{-1} \sum_{k=1}^N \beta_k \eta_{j, k}(\bd{X}). 
\end{equation}%
Finally, note that \eqref{eq:infswap} can be written as
\begin{equation}
  \label{eq:infswap2}
    \ud \bd{x}_j = - \nabla_{\bd{x}_j} \mathcal{V}(\bd{X}) \ud t +
    \sqrt{2 \beta^{-1}} \ud \bd{W}_t^{(j)},
\end{equation}
where we defined the mixture potential
\begin{equation}
  \label{eq:26}
  \mathcal{V}(\bd{X}) = -\beta^{-1} \log \sum_{\sigma\in P_N}
  e^{-\beta \mathcal{V}(\bd{X},\sigma)} = -\beta^{-1} \log \sum_{\sigma\in P_N}
  e^{-\sum_{j=1}^N \beta_{\sigma(j)} V(\bd{x}_j)}
\end{equation}
As a result, the convergence properties of parallel tempering in the
infinite swapping limit (and the boost in efficiency this method
provides on vanilla sampling via~\eqref{eq:overdamp}) can be analyzed
via examination of the geometrical properties
of~$\mathcal{V}(\bd{X})$. This is what we will do in the next section
on an example.

\section{Efficiency analysis in  the harmonic case}
\label{sec:harmo}

From the discussion above, we know that, given any set of temperatures,
the optimal way to operate parallel tempering is in fact to simulate
the effective equation in~\eqref{eq:infswap} that emerges in the
$\nu\to\infty$ limit. This is clearly doable to the extent that we can
calculate the factors $R_j(\bd{X})$ defined in~\eqref{eq:defRj}. Since
these factors contained sums over $\sigma \in P_N$, i.e., over $N!$
terms, calculating $R_j(\bd{X})$ by exhaustive evaluation of these
sums is only an option if the number $N$ of temperatures remains
moderately low. In practice, however, one often needs to use
$N = 50 - 100$, for which these exhaustive calculations are not
doable. How to proceed in these situations will be explained in
Section~\ref{sec:hmm}. The question we address here is why there is a
practical need to use a rather larger number of temperatures rather
than, say, $N=2$. This will also allow us to investigate how to choose the
temperature ladder in the method.

Insight about this issue can be gained by looking at the special case
when the potential is harmonic
\begin{equation}
  \label{eq:harmonic}
  V(\bd{x}) =  \tfrac{1}{2} \abs{\bd{x}}^2, \qquad \bd{x} \in \RR^D,
\end{equation}
in which case the canonical density reduces to
\begin{equation}
  \label{eq:11}
  \rho_\beta(\bd{x}) = \bigl(2\pi\beta^{-1}\bigr)^{-\frac{D}2} e^{-\frac12 \beta |\bd{x}|^2}.
\end{equation}
Even though this example is very simple, it has the advantage to be
amenable to analysis and it illustrates difficulties that are shared
with more realistic (and complex) systems.  We can always set
$\beta=1$ by rescaling $\bd{x}\to\beta^{-1} \bd{x}$, so let us focus
on this case here. As we have seen in Proposition~\ref{prop:variance},
the variance of the estimator for parallel tempering can be controlled
by the variance of the observable, and thus we shall just focus on the
convergence of the dynamics to the invariant measure.

% In this setup, for an observable $f(\bd{x})$, The vanilla estimator  is
% %
% \begin{equation}
%   \label{eq:23}
%   \average{f}_T = \frac1{T} \int_0^T f(\bd{x}(t))\ud t,
% \end{equation}
% %
% where  $\bd{x}(t)$ is the solution of the overdamped Langevin
% equation associated with the potential~\eqref{eq:harmonic}:
% %
% \begin{equation}
%   \label{eq:12}
%   \ud \bd{x} = - \nabla V(\bd{x}) \ud t + \sqrt{2} \ud \bd{W} = 
%   - \bd{x} \ud t + \sqrt{2} \ud \bd{W}.
% \end{equation}
% %
% Indeed, the ergodic theorem guarantees that
% $\average{f}_T \to \average{f}$ almost surely as $T\to\infty$. To
% assess the rate of convergence, we calculate the variance of the
% estimator~\eqref{eq:23}. Using the fact that correlation time for
% \eqref{eq:12} is inverse proportional to the smallest non-zero
% eigenvalue of the generator $\lambda_1 = 2$, we obtain the estimate
% \cite{BrooksHandbook:11}
% %
% \begin{equation}
%   \label{eq:24}
%   \var\bigl( P^T_D(a)\bigr) \le \frac{P_D(a)(1-P_D(a))}{\lambda_1 T
%   }\sim \frac{P_D(a)}{\lambda_1 T } \qquad \text{as} \quad D \to\infty.
% \end{equation}
% %
% and we expect that this bound will be rather sharp in general.  This
% means that, in this limit, the efficient of the
% estimator~\eqref{eq:23} deteriorates badly, since its relative error
% scales as
% %
% \begin{equation}
%   \label{eq:25}
%   \frac{\sqrt{\var\left( P^T_D(a)\right) }}{P_D(a)} \lesssim
%   \frac{1}{\sqrt{\lambda_1 T P_D(a)} } \qquad \text{as}\quad D \to\infty.
% \end{equation}

\subsection{The case with two temperatures}
\label{sec:2temp}

Let us first consider the situation where we only add one artificial
temperature, so that we have two replicas with
$\beta_1=\beta = 1$ and $\beta_2=\bar{\beta} <1$. In this case, the
effective equation we obtain in the $\nu\to\infty$ limit reads
\begin{equation}
  \begin{aligned}
    & \ud \bd{x}_1 = - r_1(V(\bd{x}_1), V(\bd{x}_2)) \nabla
    V(\bd{x}_1) \ud t + \sqrt{2} \ud W_1
    =- r_1\left(\tfrac12|\bd{x}_1|^2, \tfrac12|\bd{x}_2|^2\right)
    \bd{x}_1\ud t + \sqrt{2} \ud W_1 \\
    & \ud \bd{x}_2 = - r_2(V(\bd{x}_1), V(\bd{x}_2)) \nabla V(\bd{x}_2) \ud
    t + \sqrt{2} \ud W_2 = - r_2\left(\tfrac12|\bd{x}_1|^2,
      \tfrac12|\bd{x}_2|^2\right) \bd{x}_2 \ud t + \sqrt{2} \ud W_2
\end{aligned}\label{eq:3}
\end{equation}
where 
\begin{equation}
  r_1(E_1, E_2) = \dfrac{e^{-E_1 - \bar{\beta} E_2} + \bar{\beta}
    e^{-E_2 - \bar{\beta} E_1}}{e^{-E_1 - \bar{\beta} E_2} + e^{-E_2 -
      \bar{\beta} E_1}}, \qquad r_2(E_1, E_2) = r_1(E_2,E_1)
\end{equation}
Therefore, we can write down a closed evolution equation for the law
of $\mathcal{E}_1 = \frac1{2D}|\bd{x}_1|^2$ and
$\mathcal{E}_2 = \frac1{2D}|\bd{x}_2|^2$. A few simple manipulations
show that this equation can be written as
\begin{equation}
  \label{eq:1}
  \ud\begin{pmatrix}
    \mathcal{E}_1 \\
    \mathcal{E}_2
  \end{pmatrix} = \begin{pmatrix}
    2\mathcal{E}_1 & 0 \\
    0& 2\mathcal{E}_2
  \end{pmatrix} \begin{pmatrix}
    \partial_{\mathcal{E}_1} \log g(\mathcal{E}_1,\mathcal{E}_2)\\
    \partial_{\mathcal{E}_2}\log g(\mathcal{E}_1,\mathcal{E}_2)
    \end{pmatrix}
    \ud t + D^{-1}\begin{pmatrix}
    2\\
    2
  \end{pmatrix}  \ud t + \sqrt{2D^{-1}} \begin{pmatrix}
    \sqrt{2\mathcal{E}_1 }& 0 \\
    0& \sqrt{2\mathcal{E}_2}
  \end{pmatrix} \begin{pmatrix}
    \ud W_1\\
    \ud W_2
    \end{pmatrix},
\end{equation}
where 
\begin{equation}
  \label{eq:5}
  g(\mathcal{E}_1,\mathcal{E}_2) =
  (\mathcal{E}_1\mathcal{E}_2)^{\frac12-\frac{1}{D}}\left(
  e^{-D(\mathcal{E}_1+ \bar{\beta}
    \mathcal{E}_2)}+e^{-D(\mathcal{E}_2+
    \bar{\beta} \mathcal{E}_1)}\right)^{1/D}.
\end{equation}
Writing compactly
$\mathcal{E}= (\mathcal{E}_1, \mathcal{E}_2)^{\top}$, this equation is of
the from
\begin{equation}
   \label{eq:2}
   \ud\mathcal{E} = M( \mathcal{E} ) \nabla_{\mathcal{E}} \log
   g(\mathcal{E}) \ud t + D^{-1}\divop\, M( \mathcal{E} ) \ud t 
   +    \sqrt{2D^{-1}} M^{1/2} ( \mathcal{E} )\ud W
\end{equation}
with $M( \mathcal{E} )  = \text{diag}(2 \mathcal{E}_1,2 \mathcal{E}_2
)$, which indicates that its invariant density is proportional to
$g^D$, i.e. it is given by 
\begin{equation}
  \label{eq:6}
  \varrho(\mathcal{E}_1,\mathcal{E}_2) = \mathcal{C}_D^{-1}
  (\mathcal{E}_1\mathcal{E}_2)^{\frac{D}2-1}\left(
  e^{-D(\mathcal{E}_1+\bar{\beta} \mathcal{E}_2)}+e^{-D(\mathcal{E}_2+ \bar{\beta} \mathcal{E}_1)}\right).
\end{equation}
where the normalization constant is given by 
\begin{equation*}
  \begin{aligned}
    \mathcal{C}_D & = \int
    (\mathcal{E}_1\mathcal{E}_2)^{\frac{D}2-1}\left(
      e^{-D(\mathcal{E}_1+\bar{\beta} \mathcal{E}_2)}+e^{-D(\mathcal{E}_2+ \bar{\beta} \mathcal{E}_1)}\right) \ud \mathcal{E}_1 \ud \mathcal{E}_2 \\
    & = 2 D^{-D/2} (\bar\beta D)^{-D/2} \Gamma \Bigl(\frac{D}{2}\Bigr) ^2  \\
    & \sim \frac{2}{D} (2e)^{-D} \bar\beta^{-D/2}
  \end{aligned}
\end{equation*}
as $D\to\infty$.  Thus, \eqref{eq:1} and equivalently \eqref{eq:2}
describe diffusion on the energy landscape
\begin{equation}
  \label{eq:7}
  -\log\, g(\mathcal{E}_1,\mathcal{E}_2).
\end{equation}
As long as $\bar \beta<1$, this landscape possesses two minima with a
saddle point in between (see Fig.~\ref{fig:remdex}).
\begin{figure}[ht]
  \centering
  \includegraphics[width = 4in]{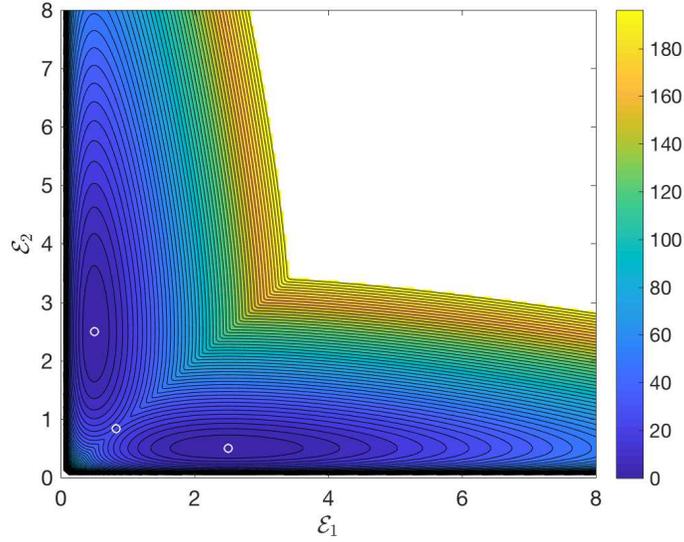}
  \caption{Energy landscape of $\mc{E}_1$ and $\mc{E}_2$ for
    $D = 100$ and $\bar{\beta} = 0.2$. The white circles show the
    location of the two minima and the saddle points, located
    approximately at~\eqref{eq:8} and~\eqref{eq:9}, respectively. The
    potential is also capped at 200 for clarity.}
  \label{fig:remdex}
\end{figure}
For large $N$, the minima
are approximately (that is, to leading order in $N^{-1}$) located at
\begin{equation}
  \label{eq:8}
  (\mathcal{E}_1,\mathcal{E}_2) = \frac12(1, \bar \beta^{-1}), \quad \text{and}
  \quad (\mathcal{E}_1,\mathcal{E}_2) = \frac12(\bar \beta^{-1},1)
\end{equation}
with energy approximately given by $E_m = 1 +\log 2 + \frac{1}{2} \log
\bar \beta$ and the saddle point is approximately located at
\begin{equation}
  \label{eq:9}
  (\mathcal{E}_1,\mathcal{E}_2 ) = \bigl((1+\bar \beta)^{-1}, (1+\bar \beta)^{-1}\bigr)
\end{equation}
with energy approximately given by $E_s = 1 +\log (1+\bar \beta)$.
Therefore, for large $D$, we can use results from large deviation
theory to deduce that the smallest nonzero eigenvalue of the generator
of~\eqref{eq:1} (that is, roughly, the inverse of mean first passage
time between the minima) is log-asymptotically given by
\begin{equation}
  \label{eq:10}
  \bar{\lambda}_1 \asymp \exp(- D \Delta E), \qquad 
  \Delta E = E_s-E_m = \log (1+\bar \beta) -\log 2 +
  \frac12 \log \bar\beta^{-1} \qquad \text{as} \quad D\to\infty.
\end{equation}
Therefore, for any $\bar \beta<1$ fixed, the system becomes
increasingly metastable as $D$ gets larger. In particular, if we start
with a configuration with $\mc{E}_1 < \mc{E}_2$, it will take
exponentially (in $D$) long time for the dynamics to reach the
configurations with $\mc{E}_1 > \mc{E}_2$. This means exponentially
slow convergence to equilibrium, and shows that using two temperature
is not sufficient in general.

\subsection{The case with many temperatures}
\label{sec:Ntemp}

Consider now what happens with multiple temperatures, in which case a
similar analysis as in the previous section can be done. In
particular, the mixture potential (for the energies) has local minima
at
\begin{equation}
\mathcal{E}_i = \frac{1}{2\beta_{\sigma_i}} \qquad \text{for any permutation }\sigma,
\end{equation}
whose energies scale as
\begin{equation}
  \frac12 \sum_i \log \beta_i + \frac{N}{2} (\log 2 +1).
\end{equation}
The saddle points are located at 
\begin{equation}
\begin{cases}
 \dps\mathcal{E}_i = \frac{1}{2\beta_{\sigma_i}} \qquad \text{for}\; i\not=j,k \\
 \mathcal{E}_j = \mathcal{E}_k = (\beta_{\sigma_j}+\beta_{\sigma_k})^{-1}
\end{cases} 
\end{equation}
for any permutation $\sigma$ and any pair $j\not=k$, whose 
energies scale as 
\begin{equation}
  \frac12 \sum_{i\not=j,k} \log \beta_j + \log (\beta_j+\beta_k) + \frac{N-2}{2} \log 2 + \frac{N}2.
\end{equation}
Therefore the barrier scale as 
\begin{equation}
  \Delta E_{j,k} = \log \biggl( \frac{\beta_j+\beta_k}{(\beta_j \beta_k)^{1/2}} \biggr) - \log 2.
\end{equation}
Assuming that $\beta_1 >\beta_2 > \ldots > \beta_N$, this indicates
that we are more likely to see transition between minima lying on
adjacent $\beta_j$s, and in order to keep all the rates equal (if we
use Arrhenius formula such that rate is given by exponential of the
energy barrier), we must have
\begin{equation}
  \frac{\beta_j+\beta_{j+1}}{(\beta_j \beta_{j+1})^{1/2}} = \text{cst}
\end{equation}
for all $j=1, \ldots, N$. This agrees with the analysis based on
transition state theory performed in our previous work \cite{Lu2013},
and shows that the optimal choice of temperature on the ladder should
satisfy
\begin{equation}
\frac{\beta_{j+1}}{\beta_j} = \Bigl(\frac{\beta_N}{\beta_1}\Bigr)^{1/(N-1)}
\end{equation}
and that this ratio should be of order $1$.  This indicates that using
multiple temperature removes the energy barrier and hence makes
convergence much faster.

\section{Simulation of the stochastic switching dynamics}
\label{sec:hmm}

\subsection{Direct simulation based on stochastic simulation
  algorithm}
\label{sec:ssa}

Here we present an exact implementation of the stochastic switching
dynamics, assuming that we can integrate the Markov process
corresponds to $\mc{L}_\sigma$ exactly. In practice, this
integration often needs to be discretized, which will introduce some additional
errors. We will discuss these later.

Given that at time $t$ the system is in the state $(\bd{X}(t), \sigma(t))$,
the algorithm updates the state via the following steps:
\begin{enumerate}[1.]
\item Draw a random number $r \sim U([0, 1])$;

\item Integrate $\bd{X}$ with the fixed $\sigma := \sigma(t)$ to time
  $t+\tau$, such that 
  \begin{equation*}
    r = \exp\biggr( - \nu \sum_{ \sigma'\neq \sigma } \int_t^{t + \tau} 
     h_{\sigma\sigma'}(\bd{X}(s))\, \ud s 
    \biggr), 
  \end{equation*}
  where $r$ is the random number drawn in Step 1;
\item Choose $\sigma(t + \tau) $ with probability 
  \begin{equation*}
    \PP\bigl( \sigma(t+\tau) = \sigma' \mid \sigma, \tau \bigr)  
    = \frac{h_{\sigma\sigma'}(\bd{X}(t+\tau))}
    {\dps \sum_{\wt{\sigma} \neq \sigma} h_{\sigma\wt{\sigma}}(\bd{X}(t+\tau))};
  \end{equation*}
\item Update the time to $t + \tau$ and repeat till the target time.
\end{enumerate}

If a large swapping rate $\nu$ is employed, the time lags $\tau$ will
likely be rather small as their expectation is proportional to
$\nu^{-1}$. In these situations, the stochastic switching dynamics
exhibits two scales: the time scale of the diffusion process and that
of the jumping process which is much faster. As a result, direct
simulation by the algorithm discussed above is no longer
efficient. This difficulty is in fact very similar to the one we
encounter when numerically integrating SDEs with multiple time scales,
for which a by-now standard idea is to explore the asymptotic limit
explicitly to design a more efficient multiscale integrator. In what
follows, we explore such an integrator, which was proposed in
\cite{TangLuAbramsVE} and fits the framework of the heterogeneous
multiscale methods (HMM)~\cites{EEngquist:03,EVE:03,Weinan2007}.

\subsection{A multiscale integrator} 
\label{sec:multi}

The basic idea of the HMM integrator is to compute on-the-fly the
scaling factor $R_j(\bd{X})$ in the limiting
equation~\eqref{eq:infswap} for the slow variables~$\bd{X}$ using
short bursts of simulation of the fast process $\sigma$ -- these
simulations are performed in what is called the microsolver in HMM,
whereas the routine used to evolve the slow variables is referred to
as the macrosolver. The data generated by the microsolver is passed to
the unknown coefficients for the macrosolver via an estimator.  In the
present context, the HMM scheme we will use works as follows:

Choose a time step $\Delta t$ appropriate for the simulation of the
limiting SDEs~\eqref{eq:infswap} and pick a frequency $\nu$ such that
$\nu \gg 1/\Delta t$ (for example
$\nu = 10^2/\Delta t - 10^3/\Delta t$).  Start the simulation with
$\bd{X}^{(0)}= \bd{X}(0)$ and $\sigma^{(0)}=\sigma(0)$, and set $t_0=0$. Then
for $k\ge0$, we use the following 3-step iteration algorithm to
update:\begin{enumerate}[1.]
\item \textit{Microsolver:} Evolve $\sigma^{(k)}$ via SSA from $t_k$
  to $t_{k+1} := t_k+\Delta t$ using the rate in~\eqref{eq:jumpProb}
  with $\bd{X} = \bd{X}^{(k)}$ fixed, that is: Set
  $\sigma^{(k,0)} = \sigma^{(k)}$, $t_{k,0}=t_k$, and for $l\ge1$, do:
  \begin{enumerate}
  \item Compute the lag to the next event via
    \begin{equation*}
      \tau_l = - \frac{\ln r}{\nu h_{\sigma^{(k,l-1)}}(\bd{X}^{(k)})}
    \end{equation*}
    where $r$ is a random number uniformly distributed in the interval
    $(0,1)$ and
    $h_{\sigma}(\bd{X}^{(k)}) = \sum_{\sigma'\not=\sigma}
    h_{\sigma,\sigma'}(\bd{X}^{(k)})$;
  \item Pick $\sigma_{k,l}$ with probability 
    \begin{equation*}
      \PP_{\sigma^{(k,l)}} =  \frac{h_{\sigma^{(k,l-1)},\sigma^{(k,l)}}(\bd{X}^{(k)})}{h_{\sigma^{(k,l-1)}}(\bd{X}^{(k)})};
    \end{equation*}
  \item Set $t_{k,l} = t_{k,l-1}+ \tau_l$ and repeat till the first
    $L$ such that $t_{k,L}>t_k+\Delta t$; then set $\sigma^{(k+1)}=
    \sigma^{(k,L)}$ and reset $\tau_L = t_k+\Delta t - t_{k,L-1}$.
  \end{enumerate}

\item \textit{Estimator:}  Given the trajectory of $\sigma$, 
  estimate $\eta_{j, i}(\bd{X}^{(k)})$ and $R_j(\bd{X}^{(k)})$ via
  \begin{equation*}
    \begin{aligned}
      \wh{\eta}_{j,i}(\bd{X}^{(k)}) & = \frac{1}{\Delta t}
      \int_{t_k}^{t_k+\Delta t} \mathbb{1}_{i=\sigma_j(s)} \ud s\\
      & = \frac{1}{\Delta t} \sum_{l=1}^L \tau_l
      \mathbb{1}_{i=\sigma^{(k,l)}_j} ,
    \end{aligned}
  \end{equation*}
  and
  \begin{equation*}
    \begin{aligned}
      \wh{R}_j(\bd{X}^{(k)}) & = 
      \beta^{-1} \sum_i \frac{\beta_i}{\Delta t} \int_{t_k}^{t_k+\Delta t}
      \mathbb{1}_{i=\sigma_j(s)} \ud s\\
      &= \beta^{-1} \sum_i \beta_i \,\hat{\eta}_{j,i}(\bd{X}^{(k)}).
    \end{aligned}
  \end{equation*}
\item \textit{Macrosolver:} Evolve $\bd{X}^{(k)}$ to $\bd{X}^{(k+1)}$ using one
  time-step of size $\Delta t$ in the numerical integrator
  for~\eqref{eq:infswap} with $R_j(\bd{X}^{(k)})$ replaced by the factor
  $\widehat{R}_j(\bd{X}^{(k)})$ calculated in the estimator. Then repeat the
  three steps above for each time-step $k$.
\end{enumerate}

\smallskip 

\begin{remark}
  It is useful to compare the above scheme with a time-splitting
  integrator for the dynamics of $(\bd{X}(t), \sigma(t))$, that is for
  time step size $\Delta t$, we alternate between evolving $\bd{X}(t)$
  and $\sigma(t)$ while freezing the other component. More precisely,
  from $t_k$ to $t_{k+1} = t_k + \Delta t$ the algorithm updates the
  state via
\begin{enumerate}[1.]
\item Evolve $\sigma^{(k)}$ via SSA from $t_k$ to
  $t_{k+1} := t_k+\Delta t$ using the rate in~\eqref{eq:jumpProb} with
  $\bd{X} = \bd{X}^{(k)}$ to get $\sigma^{(k+1)}$, using the similar
  procedure as the \emph{Microsolver} in the multiscale integrator;
\item Evolve $\bd{X}^{(k)}$ to $\bd{X}^{(k+1)}$ using
  $\mc{L}_{\sigma^{(k+1)}}$, namely, if Euler-Maruyama scheme is used
  \begin{equation}
    \bd{x}_j^{(k+1)} = \bd{x}_j^{(k)} - \Delta t \beta^{-1}
    \beta_{\sigma_j^{(k+1)}} 
    \nabla V(\bd{x}_j^{(k)}) + \sqrt{2 \beta^{-1} \Delta t}\; \eta_j^{(k)}, 
  \end{equation}
  where $\eta_j^{(k)}$ are i.i.d.~standard Gaussian random variables.
\end{enumerate}
The difference with the multiscale integrator lies in the evolution of
$\bd{X}$, where the splitting scheme uses the current instance
$\sigma^{(k+1)}$ while the multiscale integrator incorporates into the
average of the coefficients. Since we are interested in the regime
that the dynamics of $\sigma$ is much faster of $\bd{X}$, using the
average is preferred as it is
consistent with the $\nu \to \infty$ limit.
\end{remark}

We will not go into details of the numerical analysis of the
multiscale integrator here. Numerical analysis of HMM integrators in
similar spirit can be found in \cite{EEngquist:03, EVE:03, EVE:07,
  ELiuVE:05, TaoOwhadiMarsden2010, LuSpiliopoulos}, the basic idea is
to show that the scheme is consistent with the averaging limit as
$\nu \to \infty$ by showing that the estimator gives a good
approximation to $R_j(\bd{X})$.

\section{Extension to simulated tempering}\label{sec:simulated}

The analysis we have performed in the context of parallel tempering
can be generalized to simulated tempering, as we briefly explain in this
section. For a performance comparison between parallel and simulated
tempering, we refer the readers to \cite{Zhang2008}.

We begin by noting that simulated tempering, in which we use one
single replica but make the temperatures dynamically switch between
$N$ values $\beta_1$, $\beta_2$, \ldots, $\beta_N$, can be cast in a
framework similar to that of the stochastic switching process
discussed in Section~\ref{sec:remd}. This amounts to assuming that the
state space of the process is $\Omega^N \times \{1, \ldots, N\}$, and
requiring that this process samples the joint equilibrium measure
\begin{equation}
  \label{eq:imst}
  \pi(\ud \bd{x},i) = \varrho(\bd{x}, i) \ud \bd{x}, \qquad \text{with} \quad
  \varrho(\bd{x}, i)
  = \frac{n_i e^{-\beta_i V(\bd{x})}}{\sum_j n_j Z_{\beta_j}}. 
\end{equation}
where $\{n_i\}_{i=1,\ldots, N}$ is a set of positive weights to be
chosen \textit{a~priori}.  To this end, the dynamics in $x$ given $i$
can be chosen as a time rescaled overdamped equation:
\begin{equation}
  \ud \bd{x} (t) = - \beta^{-1} \beta_i \nabla V(\bd{x} (t)) + 
  \sqrt{2 \beta^{-1}} \ud \bd{W}_t, 
\end{equation}
so that it samples 
\begin{equation}
  \varrho(\bd{x} \mid i) = \frac{\varrho(\bd{x}, i)}{\int_{\Omega} \varrho(\bd{y}, i)
    \ud \bd{y}} = Z_{\beta_i}^{-1} e^{-\beta_i V(\bd{x})}. 
\end{equation}
The jumping rate matrix between the temperature indices $i$ given $x$
can be chosen e.g.
\begin{equation}
  \nu k_{ij}(\bd{x}) =  \nu \min\Bigl( \frac{n_{j}}{n_{i}} e^{-(\beta_j -
    \beta_i) V(\bd{x})}, 1 \Bigr), \qquad i\not = j 
\end{equation}
with $k_{ii}(\bd{x}) = -\sum_{j\not=i}k_{ij}$ so that the jumping process
samples
\begin{equation}
  \varrho(i \mid \bd{x}) = \frac{\varrho(\bd{x}, i)}{\sum_{j=1}^N
    \varrho(\bd{x}, j)} 
  = \frac{n_i e^{-\beta_i V(\bd{x})}}{\sum_j n_j e^{-\beta_j V(\bd{x})}}. 
\end{equation}
Here $\nu$ is again an overall scaling parameter for the switching
frequency.
Other choices such as
\begin{equation}
  k_{ij}(\bd{x}) =  \frac{n_j}{n_i e^{-(\beta_i - \beta_j)V(\bd{x})} + n_j}, \qquad i\not = j,
\end{equation}
or 
\begin{equation}
  k_{ij}(\bd{x}) =  \sqrt{\frac{n_j}{n_i}} e^{-\frac{1}{2} (\beta_j -
    \beta_i) V(\bd{x})}, \qquad i\not = j,
\end{equation}
are possible as well.
The infinitesimal generator of the stochastic switching 
process $(\bd{x} (t),i(t))$ defined this way is given by
\begin{equation}
  \label{eq:genst}
  (\mc{L}_{\nu} u)(\bd{x}, i) = - \beta^{-1} \beta_i \nabla V(\bd{x}) \cdot 
  \nabla u(\bd{x},\beta) 
  + \beta^{-1} \Delta u(\bd{x}, \beta) 
  + \nu \sum_{j\not=i} k_{i,j} (\bd{x})  \bigl( u(\bd{x}, j) - 
  u(\bd{x}, i) \bigr), 
\end{equation}
and it is an easy exercise to show that its invariant measure is
indeed~\eqref{eq:imst}. 

In terms of calculating expectations, in the context of simulated
tempering we can use
\begin{equation}
  \label{eq:1eavg}
  \begin{aligned}
    \average{A}_{\beta_k} &\equiv \int_{\Omega} A(\bd{x})
    \rho_{\beta_k}(\bd{x}) \ud \bd{x}\\
    & = \int_{\Omega} A(\bd{x}) g_k(\bd{x}) \varrho(x) \ud \bd{x} 
  \end{aligned}
\end{equation}
where $\varrho(\bd{x})$ is the marginal of $\varrho(\bd{x},i)$ on $\bd{x}$:
\begin{equation}
  \label{eq:4}
  \varrho(\bd{x}) = \sum_{i=1}^N \varrho(\bd{x},i) = \frac{\sum_{i=1}^N n_i
    e^{-\beta_iV(\bd{x})}}
  {\sum_{i=1}^N n_i Z_{\beta_i}},
\end{equation}
and we defined
\begin{equation}
  \label{eq:gk}
  g_k (\bd{x}) = \frac{\rho_{\beta_k}(\bd{x})}{\varrho(\bd{x})} = 
    \frac{Z_{\beta_k}^{-1}e^{-\beta_k V(\bd{x}) } 
      \sum_{j=1}^N n_j Z_{\beta_j}}{\sum_{j=1}^N n_j e^{-\beta_j V(\bd{x})}}.
\end{equation}
These expressions show the main (and well-known) difficulty one is
faced when using simulated tempering: even though it looks simpler
than simulated tempering because it involves a single replica, it
requires one to learn the partition functions $Z_{\beta_i}$ to
calculate expectations (in contrast, parallel tempering does not
require this).

Setting this difficulty inside, an argument similar to the one
presented in Section~\ref{sec:LDP} indicates that the large deviation
rate functional for the empirical measure of the process with
generator~\eqref{eq:genst} is the sum of the terms, the first of which
is independent of $\nu$ and the second proportional to $\nu$. This
indicates that, in the context of simulating tempering too, it is
optimal to take the limit as $\nu\to\infty$ (which we will refer to as
the `infinite switching limit'). In this limit, it is easy to show
using averaging arguments similar to those presented in
Section~\ref{sec:infinity}, that the process $\bd{x} (t)$ converges to the
solution of the following effective equation
\begin{equation}
  \label{eq:19}
  \ud \bd{x} (t) = - \beta^{-1} \wb{\beta}(\bd{x} (t)) \nabla V(\bd{x}
  (t)) 
  + \sqrt{2 \beta^{-1}} \ud \bd{W}_t, 
\end{equation}
where we defined
\begin{equation}
  \label{eq:20}
  \wb{\beta}(\bd{x}) = \sum_{i=1}^N \beta_i \varrho(i\mid \bd{x}) = 
  \frac{\sum_{i=1}^N \beta_i n_i e^{-\beta_i V(\bd{x})}}{\sum_{i=1}^N n_i e^{-\beta_i V(\bd{x})}}. 
\end{equation}
This equation is in principle much easier to simulate than the
corresponding effective equation~\eqref{eq:infswap} obtained in the
infinite switching limit of paralel tempering, since the sum
in~\eqref{eq:20} involves $N$ terms instead of $N!$ (and hence can be
evaluated straightforwardly even for large $N$). However, the
efficiency boost of the method depends crucially on the choice of the
weights~$n_i$. It has been argued that the optimal choice is to take
$n_i = Z_{\beta_i}^{-1}$, but, unfortunately, these partition functions
are unknown \textit{a~priori} (which brings us back to the issue with
the estimator~\eqref{eq:1eavg} for the expectations). How to go around
these difficulties (and justify the optimality of the choice
$n_i = Z_{\beta_i}^{-1}$ starting from the effective equation
in~\eqref{eq:19}) will be discussed elsewhere.

\section{Concluding remarks}
\label{sec:conclu}

We have shown how a formulation of parallel tempering as a stochastic
switching process for the coupled dynamics of replica configuration
and temperature permutation naturally leads to considering these
equations in the infinite swapping limit. Indeed this can be justified
from the additive structure of the rate functional for the empirical
measure of the process, using tools from large deviation theory as
originally proposed in~\cite{Plattner2011} and further justified
in~\cite{DupuisLiu:2015}. This observation suggests to simulate
directly the effective equation that emerges in this limit
. Unfortunately, this is no trivial task because the direct
calculation of the coefficients in this equation involves sums over
all possible permutations of the temperatures. Clearly this is only
doable if the number of temperatures/replicas is small, which is a
set-up in which parallel tempering leads to no significant boost in
efficiency. This observation is common knowledge for practitioners of
parallel tempering, and is usually justified by showing that many
temperatures are required to get a significant acceptance rate of the
temperature swaps. Within the infinite swapping limit, we arrive at
the same conclusion, but from a different viewpoint: as we showed
here, it is necessary to use many temperatures in order to eliminate
the free energy barriers on the mixture potential.

Once one realizes that the effective equation emerging in the infinite
swapping limit must be operated with many temperatures, the question
arises as how to do so in practice. Following the suggestion made
in~\cite{TangLuAbramsVE}, we showed here that it can be done using
multiscale integrators such as HMM. These schemes are straightforward
generalization of those traditionally used in parallel tempering, and
we therefore they should be useful to improve the efficiency of this
method with rather minimal modifications of the existing
codes. Alternatively, we showed that one can use similar ideas in the
context of simulated tempering, in which case the limiting equation
that arise in the infinite switching limit is simpler.

\section*{Acknowledgments}

We thank C. Abrams, B. Leimkuhler, and A. Martinsson for interesting
discussions. The work of JL is supported in part by the National
Science Foundation (NSF) under grant DMS-1454939. The work of EVE is
supported in part by the Materials Research Science and Engineering
Center (MRSEC) program of the NSF under award number DMR-1420073 and
by NSF under award number DMS-1522767.

\bibliographystyle{amsxport}
\bibliography{ssa}

\end{document}